\documentclass{article}

\usepackage{arxiv}
\usepackage{algorithm}
\usepackage{algorithmicx}
\usepackage[noend]{algpseudocode}
\usepackage[utf8]{inputenc} 
\usepackage[T1]{fontenc}    
\usepackage{hyperref}       
\usepackage{url}            
\usepackage{booktabs}       
\usepackage{amsfonts}       
\usepackage{nicefrac}       
\usepackage{microtype}      
\usepackage{lipsum}
\usepackage{amsmath,amsfonts,amssymb,amsthm,epsfig,epstopdf,titling,url,array}

\usepackage[dvipsnames]{xcolor}

\usepackage{setspace}

\newtheorem{theorem}{Theorem}
\newtheorem{lemma}{Lemma}
\theoremstyle{definition}
\newtheorem{definition}{Definition}
\newtheorem{example}{Example}

\title{SURFACE: A Practical Blockchain Consensus Algorithm for Real-World Networks}

\author{Zhijie Ren \;and\; Ziheng Zhou \\
  VeChain\\
  \texttt{zhijie.ren, peter.zhou@vechain.com} \\
}

\begin{document}
\maketitle

\begin{abstract}
SURFACE, standing for \textbf{S}ecure, \textbf{U}se-case adaptive, and \textbf{R}elatively \textbf{F}ork-free \textbf{A}pproach of \textbf{C}hain \textbf{E}xtension, is a consensus algorithm that is designed for real-world networks and enjoys the benefits from both the Nakamoto consensus and Byzantine Fault Tolerance (BFT) consensus. In SURFACE, a committee is randomly selected every round to validate and endorse the proposed new block. The size of the committee can be adjusted according to the underlying network to make the blockchain mostly fork-free with a reasonable overhead in communication. Consequently, the blockchain can normally achieve fast probabilistic confirmation with high throughput and low latency. SURFACE also provides a BFT mechanism to guarantee ledger consistency in case of an extreme  network situation such as large network partition or being under a massive DDoS attack.
\end{abstract}


\section{Introduction}
Blockchain, originated from Bitcoin \cite{nakamoto}, has received great attention recently since it can be used to create a trusted ledger/system amongst multiple untrusted parties without a central authority or an authorized third party. One of the mostly discussed and vastly studied problems in blockchain is improving the throughput and latency, or in particular, the Bitcoin POW scheme, which is called the scalability problem \cite{croman}. Briefly speaking, the traditional Bitcoin POW is sub-optimal due to various reasons, one of which is the dependency between the Bitcoin POW's consensus and the synchrony of the network \cite{croman,wattenhofer}. More precisely, forks would be created due to the asynchrony of the network, which waste bandwidth to be transmitted, reduce the security, and thus reduces the throughput and increases the confirmation time, i.e., the latency. One way to mitigate this problem is to transform the chain structure to a Directed Acyclic Graph (DAG) so that all branches of blocks could be counted in the consensus. As a result, the mining power as well as the transactions on the forks are all taken into account and thus the throughput and the security would not be degraded \cite{spectre,phantom,conflux,prism}. Another approach is to prevent forks, i.e., using Byzantine Fault Tolerant~(BFT) algorithms to reach consensus for each block \cite{hybrid,byzcoin,algorand}.


In this paper, we introduce SURFACE, which is the abbreviation of \textbf{S}ecure, \textbf{U}se-case adaptive, and \textbf{R}elatively \textbf{F}ork-free \textbf{A}pproach of \textbf{C}hain \textbf{E}xtension. The name suggests some of the major features of SURFACE, one of which is ``relatively folk-free''. In SURFACE, we focus on practical networks in real-world and propose a blockchain consensus algorithm that is ``mostly'' fork-free by introducing a committee to validate the proposed block. At the meantime, the validation process of the committee is much simpler than running BFT algorithms in the whole network, thus it causes lower latency and communication overhead. As a result, the bandwidth will not be wasted due to transmitting forks that are eventually discarded and the latency will not be increased comparing to the schemes without committee. Moreover, the performance is adaptive as we could adjust the committee size according to the network condition to achieve the best throughput/latency performance for specific use cases.
In addition to the committee based block validation scheme, a BFT scheme is also included in SURFACE to achieve finality, i.e., uncompromised consistency (security), even under the strict asynchronous assumption.

\subsection{Background}
Whether forks are allowed in the blockchain is in fact depending on the types of consensus algorithms, i.e., whether it achieves Nakamoto consensus \cite{garay,metric} or BFT consensus \cite{lamport,bracha}. Bitcoin and Ethereum \cite{ethereum}, as well as many blockchains following their paths \cite{vechain,parity,ouroboros} achieve Nakamoto consensus. The achieved consistency is probabilistic, i.e., the probability that a block is immutable is not definitive but increases exponentially as more blocks appending to it. Some other blockchains use BFT algorithms to achieve deterministic consistency~\cite{algorand,libra,dfinity,stellar}, that is, blocks can be proven to be definitively immutable once they reach consensus.

Blockchain consensus algorithms could also be categorized into permissioned and permissionless by whether some permission from a trusted source is required for a node to participate in the consensus process. For example, Bitcoin's POW is a permissionless consensus scheme as participants only need to solve a hash function based Proof-of-Work (PoW) puzzle to participate in the consensus process without needing permissions from any party. On the other hand, EOS~\cite{eos}, using the Delegated-Proof-of-Stake (DPoS) as its consensus scheme, is permissioned. In DPoS, all stakeholders vote for a certain number of ``super nodes'' to participate in the consensus. However, eventually, the elected super nodes are permissioned by the trusted authority who initiated the vote.


\subsection{State-of-the-art}
In this paper, we focus on the consensus algorithm for permissioned blockchain that functions in large real-world networks, some of the favorable choices are BFT algorithms and Proof-of-Authority (PoA) based consensus schemes. Alternatively, we could modify permissionless consensus algorithms to our scenario.

BFT algorithms are originated from \cite{lamport} and has already been developed for decades \cite{bracha,benor,castro} before the blockchain era. However, these algorithms are either too theoretical to be used in practice or not designed for scenarios of blockchains and have poor scalability in large networks. Some novel BFT algorithms designed for blockchains like Byzcoin \cite{byzcoin}, Tendermint \cite{tendermint}, Algorand \cite{algorand}, and HotStuff BFT \cite{hotstuff} reduce the message complexity of classical BFT algorithms. However, the block interval has to be set in such a way that the consensus could be reached via multiple rounds of communication, or at least the messages from the super majority (more than 2/3 of the population) could be received, which brings a relatively high latency than PoA.


PoA is a type of permissioned Nakamoto consensus schemes, which is used in several blockchains including VeChain and Parity. In general, PoA schemes enjoy the advantage of a simpler mechanism and lower latency as the block interval only needs to be sufficient for the block to be broadcast. However, it also has many drawbacks. First, probabilistic consensus is sometimes not sufficient for practical uses, in particular, for high sensitive data or law related affairs. These is also known as the ``finality'' problem, which has been addressed in Casper \cite{casper}, Polkadots \cite{polkadots}, and some others. However, as all these algorithms are ongoing works that are independent from our research, we will not address them here but compare SURFACE to them in Section~\ref{s:comp}.

Second, as most PoA schemes are roughly a natural extension of PoW in the permissioned setting, they also suffer from the above mentioned scalability problem. This problem has already been addressed by enormous number of works with various approaches \cite{ghost,ng,ouroboros,phantom}, most of which consider permissionless settings but could be modified to be applied in permissioned blockchains. Directed Acyclic Graph (DAG) based approaches \cite{phantom,conflux} have a rather high complexity, in particular, while dealing with problems of block ordering and transaction repetition. Then, ``non-DAG'' Nakamoto consensus based approaches \cite{ng,ncmax,ouroboros,snowwhite} have their dependencies on synchrony and the performances of these algorithms are thus more situational. More precisely, in BFT based blockchain consensus algorithms, the block size and intervals have to be adjusted such that the block could be broadcast and the response from the super majority could be received in each round. However, in these Nakamoto consensus algorithms, the block size and intervals only need to be suffice for the block to be broadcast. Hence, the throughput would be higher when the message delays are small and most nodes behave honestly. However, it is not safe to only consider the best scenario as otherwise malicious nodes could easily attack the blockchain by creating forks. Hence, redundancy has to be introduced in the throughput and latency to cope with the non-ideal situations, which leads to a sub-optimal performance.



\subsection{Contribution}

SURFACE focuses on the practicality of the consensus algorithms in real-world networks and enjoys the benefit from both worlds of Nakamoto consensus and BFT consensus. More precisely, SURFACE achieves fast probabilistic confirmation with high throughput and low latency that are comparable to the state-of-the-art algorithms with Nakamoto consensus like \cite{ncmax,ouroboros}. At the meantime, SURFACE could achieve deterministic consistency (finality) as BFT algorithms like PBFT and HotStuff BFT \cite{castro,hotstuff}, and the consistency will not be compromised even if the network is asynchronous. In particular, our contributions are the following.
\begin{itemize}
    \item We use a committee endorsing mechanism for each proposed block, where the size of the committee is adjustable to the network condition in such a way that the probability of forks is minimum in the ``normal'' situation of the practical network.
    \begin{itemize}
    \item Comparing to other Nakamoto consensus algorithms without committee \cite{ng,ouroboros,snowwhite}, we minimize the probability of fork and thus achieve a higher and more stable performance in terms of throughput and latency with minimum overhead in communication.
    \item Comparing to BFT based algorithms, in case that the network is large, it does not take multiple rounds of messages responding from the super majority to reach consensus. Instead, a block could be probabilistically confirmed by a certain number of successive leaders and committees. At the meantime, the block interval could be set shorter as only a fraction of nodes need to respond. Hence, for transactions that do not require definitive consistency, SURFACE allows faster confirmation in large network.
    \end{itemize}
    \item We use a BFT based consensus algorithm to achieve deterministic consistency (finality) in asynchronous networks.
    \begin{itemize}
    \item Most Nakamoto consensus algorithms use synchronous assumptions and could not achieve consistency if the network is asynchronous, e.g., the security is not guaranteed in extreme situations. In that perspective, SURFACE achieves higher security. 
    \item We use an approach that is inspired by HotStuff BFT and prove that the consistency is achieved without any synchronous argument. As a result, instead of guaranteeing both liveness and consistency with some synchronous assumptions \cite{tendermint,casper}, we guarantee uncompromised consistency but not liveness in asynchronous network. We will argue that this is practical in real-world applications in Section~\ref{s:network}.
    \end{itemize} 
    \item We further incorporate with several novel ideas and mechanisms like block decomposition and delayed validation to further improve the performance of the consensus algorithm in practical networks.
\end{itemize}

\subsection{Outline}
In Section~\ref{s:network}, we explain practical asynchronous networks with their network and security assumptions that will be used throughout this paper. A detailed explanation of SURFACE will be given in Section~\ref{s:SURFACE} with a high level overview followed by all the functions used in SURFACE. Then, in Section~\ref{s:mec}, we explain some of the novel ideas and mechanisms used in the design of SURFACE and compare them to previous works. In Section~\ref{s:pa}, we give numerical analysis on the security of probabilitic confirmation and theoretical proofs for the consistency and liveness of SURFACE. Then, in Section~\ref{s:comp}, we compare SURFACE to some closely related and recent works. At last, we conclude our paper in Section~\ref{s:conc}.

\section{Network and Security Model}\label{s:network}

Bitcoin was described and considered as a secure value-transfer system that functions in asynchronous network as long as the majority (more than 50\%) of the hashing power is rational. However, this description has already been challenged and proven wrong in many aspects. In \cite{majority}, selfish mining is introduced, which could attack Bitcoin with merely 25\% of the hashing power. In \cite{croman,wattenhofer}, it is stated that Bitcoin is only secure if the network of Bitcoin is synchronous. Then, in \cite{rational}, a collection of literature pointing out vulnerabilities in Bitcoin is listed \cite{redballoon,minerdilemma} etc., which leads to a conclusion that ``rationality'' is not a potent argument for honesty. In other words, evidences show that it is very hard to design incentive mechanism such that the rational players would only perform a certain behaviors. Hence, rational players should be considered as Byzantine players who would behave arbitrarily.

On the other hand, despite of all vulnerabilities and synchronous limitations, Bitcoin is still considered as secure and suitable for asynchronous network in the ``common belief'' as it functions well in real-world for years. 
Hence, in this paper, we take both the theoretical asynchoronous network model and the practical network in real life into account.

We consider two scenarios: firstly, in a ``normal'' situation, we assume that the network is synchronous, the nodes are semi-trusted as their behavior is restricted by rational arguments, and the messages are propagated via gossip protocol. We argue that this is a reasonable assumption as most blockchain systems, including Bitcoin, function in this kind of networks in the practice. Then, we consider an ``abnormal'' situation that either the network is partitioned due to accident or attack, or the adversaries are trying to create inconsistency. In the ``abnormal'' situation, we use the strict asynchronous BFT assumption, i.e., the message delay as well as the behavior of the malicious nodes are arbitrary.

We assume that the network could arbitrarily switch in between these two situations. However, we assume that the abnormal situation is ``temporary'', i.e., after each abnormal period of arbitrary length, there will always be a long normal period in which the transactions could be confirmed.

Then, we aim for different goals in different scenarios. In ``normal'' situation, we aim for high throughput, fast confirmation, and probabilistic consistency which is secure in the same fashion as Bitcoin. Then, to cope with the abnormal situation, we aim for uncompromised security such that the security of the blockchain is guaranteed without any synchronous argument. However, we do not guarantee liveness in abnormal situations.

The reason of choosing this model over traditional BFT network model or Nakamoto consensus model is purely practical. First, it is reasonable to assume that the ``normal'' situation is dominant as by our observation, most blockchains are working in the normal situations and rarely experienced abnormal situations, particularly, for permissioned blockchains in which the nodes are authorized to participate in the first place. At the meantime, it is crucial that the blockchain could guarantee its consistency without using any argument of synchrony or rationality, otherwise it is vulnerable to various types of attacks.


Here, we further specify our {\em practical asynchronous network} model.

\subsection{Network and cryptographic primitives}
We consider a network with $n$ nodes, denoted by node $i\in \{1,2,\ldots,n\}$. The number of adversaries is $f$ that satisfying $3f+1 \leq n$. We assume that nodes have limited computation capacity such that they could not break the cryptographic primitives used in this paper. The hash function is modelled as a random oracle.

\subsection{Normal Situation}

In the normal situation, the message delay is upper bounded by a known $\Delta$. Moreover, messages are propagated to the network via gossip protocol, i.e., adversaries cannot secretly split the network by sending different messages to different nodes, as these messages will eventually be gossiped to all nodes. W.l.o.g., we assume that in synchronous situations, if a message is received by an honest node at time $t$, then it will be received by all honest nodes before $t+\Delta$.

Moreover, in normal situation, we assume there is an incentive or punishment mechanism so that adversaries are reluctant to perform attacks on liveness, i.e., they will respond and propose blocks normally.


\subsection{Abnormal Situation}


In abnormal situation, the network is asynchronous but not hostile. By not hostile we mean that the duration of message delay follows some random distribution which is independent of the message content. In particular, as we use a VRF based committee selection scheme, the role of a committee member is not revealed until they send their proofs. Hence, a non-hostile asynchronous network will not ``coincidently'' delay all messages from nodes with a certain roles, as if the network could break the VRF or know the private keys of the nodes in advance. 

The malicious nodes are Byzantine nodes, i.e., they could perform arbitrary behaviors, including not responding or behaving honestly. We also allow them to manipulate the network, e.g., arbitrarily delay the message. Package losses, i.e., the losses of messages send by honest nodes, are not considered in this paper.

\subsection{Situations changes}

The network could change between these two situations arbitrarily and the duration of any situation is arbitrary. However, it is guaranteed that abnormal situation is temporary, i.e., after each period of abnormal situation, there will be a relatively ``long'' period of normal situation, i.e., the period is sufficiently long for some transactions to be confirmed. Then, we assume that when the situation is back to normal, all blocks that are transmitted by honest nodes during the abnormal situation will be received by all honest nodes.

\subsection{Differences from other models}

Our network assumption is a strictly stronger assumption than asynchronous assumption and partial synchronous assumption, as we make extra assumptions of the ``normal situation'' and the ``asynchronous but not hostile'' network. Hence, BFT algorithms that works in asynchronous networks \cite{miller} and partial synchronous networks \cite{castro,hotstuff} could also be applied in our network.

Then, it is strictly weaker than the synchronous assumption, as there could be asynchronous periods. Hence, BFT algorithms, as well as most NC algorithms, that function under synchronous assumption, fail in this network.

\subsection{Blockchain consensus algorithm in practical asynchronous networks}
In practical asynchronous networks, we propose a novel framework of blockchain consensus algorithms which contains two components:

\begin{itemize}
    \item The first component guarantees that in normal situation, the consensus algorithm should be able to achieve NC, which suggests that:
    \begin{itemize}
        \item {\bf Probabilistic consistency (safety)}: If an honest node {\em confirms} a block $B$ at height $h$, then the probability that another honest node {\em confirms} a block $B' \ne B$ at height $h$ is smaller than $\epsilon_1$.
        \item {\bf Liveness}: New transactions can be confirmed.
    \end{itemize}
    \item The second component guarantees that in abnormal situation, the consensus algorithm should be able to achieve {\bf definitive consistency (safety)}, i.e., if an honest node {\em finalizes} a block $B$ at height $h$, then there cannot be another honest node {\em finalizes} a block $B' \ne B$ at height $h$.
\end{itemize}

\section{SURFACE}\label{s:SURFACE}

In SURFACE, we merge several schemes and concepts from existing blockchain and BFT algorithms. At the meantime, we make many improvements by proposing some additional mechanisms. In this section, we first give a high level overview on SURFACE. Then, we explain SURFACE in detail by firstly giving two core functions: a main function that is executed at the beginning of each round and a block receiving function that is used every time a new correct block is received. Then, we explain all the functions and variables used in these functions, without giving much explanation of their purposes, which we will do in Section~\ref{s:mec}.

\subsection{A high level overview}
SURFACE works in practical asynchronous networks, in which we define round according to Global Stabilization Time (GST) with a predetermined time interval. Further, we define epoch as a relatively long duration consisting many rounds. In this paper, the hash function is modeled by a random oracle.

Then, for each round, we randomly select a leader to create a block and a committee for validation. The leader is selected by a hash function and a random beacon determined by some randomness created and acknowledged by the end of the previous epoch by all nodes. The committee members of a round are determined by each node comparing a random number generated by a Verifiable Random Function (VRF) \cite{vrf1} to a given threshold. As a result, the leader of a round is known by all nodes by the end of the previous epoch, and the committee members of a round are only revealed when they announce their roles by revealing their proofs.

Each leader needs to first broadcast the new block and collect the endorsements from $d$ committee members to generate a valid block. There remains a possibility of forks, although the chance is much lower than the algorithms without committee and it can only be caused by an adversarial leader colluding with an adversarial committee. In case of forks, nodes will determine the valid chain by the classical ``longest chain rule''. The above-mentioned algorithm is designed to have fast confirmation and high throughput in the normal situation.

The abnormal mode of our algorithm is a BFT consensus algorithm that guarantees consistency in asynchronous network and liveness if the network regains synchrony. We introduce a ``finality vector'' that is included in each block. It is a collection of new view messages, prepare messages, pre-commit messages, and commit messages that are sent by the leader and committee together. Then, we use a chain selection and finality rule inspired by HotStuff BFT \cite{hotstuff} to reach BFT, which is alternatively called ``finality'' in this paper, for the blockchain.


\subsection{Main Function}

At round $r$, a node $u$ calls the main function ${\tt MainFunction}(u,r,mode(r-1),F(u,r-1))$.

\begin{algorithm}
\caption{Main Function ${\tt MainFunction}(u,r,mode(r-1),F(u,r-1))$}\label{alg:np}
\begin{algorithmic}
\State $C_0(u,r-1)$ is the finalized chain till the last round
\State ${\cal TX}_0(u) \gets$ the set for all unpublished valid transactions known to $u$
\State $mode(u,r) \gets {\tt SwitchMode}(u,r,mode(r-1))$
\State ${\cal B}(u,r) \gets {\tt CandidateBlocks}({\cal B}_{hon}(u,r),{\cal B}_{val}(u,r),mode(u,r))$
\State $C(u,r) \gets {\tt Canonical Chain}({\cal B}(u,r),F(u,r-1))$
\State $F(u,r) \gets {\tt Finality}({\cal B}(u,r),C(u,r),F(u,r-1))$
\State $role(u,r) \gets {\tt RoleDetermine}(u,r,C(u,r))$
\If{$role(u,r) == {\tt 'leader'}$}
    \State ${\cal TX}(B(u,r)) \gets$ an ordered set of new transactions packed by $u$ to be published in this round
    \State $(s(B(u,r)), Sig_u(s(B(u,r)))) \gets {\tt BlockSummary}(C(u,r),{\cal TX}(B(u,r)))$
    \State Broadcast $s(B(u,r)), Sig_u(s(B(u,r)))$ and ${\cal TX}(B(u,r))$
    \State Receive for $d$ endorsements $E_v(Sig_u(s(B(u,r)))), role(v) = {\tt 'committee'}$ and combine it into a $\hat{\cal E}(B(u,r))$
    \State Broadcast $\hat{\cal E}(B(u,r)),Sig_u(\hat{\cal E}(B(u,r)))$
\ElsIf{$role(u,r)== {\tt 'committee'}$}
    \State Receive $s'$ and $sig'$
    \State $valbs(s') \gets {\tt ValBlockSum}(s',sig',C(u,r))$
    \If{$valbs(s')=={\tt 'valid'}$}
        \State $E_u(s') \gets {\tt Endorsement}(s',u)$
        \State Send $s',sig',End_u(s')$ to the leader
        \State Receive ${\cal TX}$
    \EndIf
\EndIf
\end{algorithmic}
\end{algorithm}

A node $u$ will always keep track of an unpublished valid transaction set ${\cal TX}_0(u)$. At the start of each round, it detects its mode $mode(u,r)$ of this round and determines the set of blocks that he will consider to determine the canonical chain $C(u,r)$, namely the candidate blocks ${\cal B}(u,r)$, according to the mode of this round. Then, he will update his own finality vector $F(u,r)$ according to his finality vector of the previous round as well as the finality vectors (actually the pruned finality vectors, which will be explained later) collected in the canonical chain $C(u,r)$.

Then, at the beginning of each round, each node determines his role by ${\tt RoleDetermine}(u,r)$. 
If $u$ is the leader of this round, he makes a block summary $s(B(u,r))$ according to $C(u,r)$ and ${\cal TX}(B(u,r)) \subseteq {\cal TX}_0(u)$. He broadcasts the block summary $s(B(u,r))$, a corresponding signature $Sig_u(B(u,r))$, along with ${\cal TX}(B(u,r))$ and waits for the endorsements from the committee members of this round. Once $d$ endorsements have been received, he combines these endorsements into a collected endorsements $\hat{\cal E}(B(u,r))$ and broadcasts it with $Sig_u(\hat{\cal E}(B(u,r)))$.

Then, if $u$ is a committee member of this round, he waits for the block summary $s'$ and a signature $sig'$. He validates the block summary and endorses for it by signing it if it is valid. Then, he sends his endorsement $E_u(s',sig')$ alongside with $s',sig'$ and at the meantime receives ${\cal TX}(B(u,r))$.

\subsection{Blocks}
In SURFACE, the blocks are not as straightforward as they are in Bitcoin since a block is decomposed into multiple parts that are not sent together as an actual ``block'', but separately. The matching parts are combined into a block by the receiver. Then, the receivers will validate these blocks and for each round, they will select a chain according to the {\em canonical chain} rule.

However, not all blocks are directly considered as {\em candidate blocks} for the canonical chain. Nodes will also use their local knowledge to judge whether the block is {\em suspicious}, i.e., sent by adversaries to intentionally cause inconsistency. When the node is in the normal mode, inconsistent blocks sent by the same leader or blocks received outside of their rounds are all suspicious and will not be considered as candidate blocks. However, when the node is in abnormal mode, these blocks will be also considered. Hence, nodes will keep two sets of blocks, which are called {\em valid blocks} and {\em honest blocks}, respectively. Then, they will select the candidate blocks from these two sets accordingly to the mode.

\subsubsection{Composition of blocks}
A block $B(u,r)$ consists of three parts: A set of transactions ${\cal TX}(B(u,r))$, a signed block summary $s(B(u,r)),Sig_u(s(B(u,r)))$, and a signed collected endorsement $\hat{\cal E}(B(u,r)),Sig_u(\hat{\cal E}(B(u,r)))$.

\paragraph{Transaction set ${\cal TX}(B(u,r))$}
Each node will maintain a set of valid and unpublished transactions ${\cal TX}_0(u)$ according to the canonical chain that they observed. Then, if he is in turn as a leader, he will make a transaction set to publish in this turn ${\cal TX}(B(u,r)) \subseteq {\cal TX}_0(u)$.

\paragraph{Block summary}
A block summary is composed by the following items in the exact order.
\begin{enumerate}
    \item A hash of the block summary and the collected endorsement of the previous block, i.e., $H(s(B(u,r)) | \hat{\cal E}(B(u,r)))$, where $B(u,r)$ is the last block of the canonical chain $C(u,r)$.
    \item The current epoch and the round number.
    \item A Merkle root of ${\cal TX}(B(u,r))$.
    \item A pruned finality vector $F_p(u,r)$ (Subsection~\ref{ss:final}).
\end{enumerate}


\paragraph{Collected endorsement}
Once $d$ endorsements are collected, they are ordered according to the ascending order of the public keys of committee members and concatenated into the collected endorsement $\hat{\cal E}(B(u,r))$.

\subsubsection{Valid blocks}

We consider a block $B$ as ``valid'' if it is self-contained, i.e., it is consistent with itself and all previous blocks on its chain $C$. Hence, we define the correctness of a block as the following
\begin{itemize}
    \item The leader, round number, epoch number, and the finality vector are consistent with what can be computed from $C$.
    \item The block summary and the collected endorsement are correctly signed by the leader.
    \item The Merkle root in the block summary is consistent with the transaction set.
    \item The transactions in the transaction set are all valid with regards to the chain that it is on.
    \item There are $d$ endorsements in the collected endorsements from the committee members of round $r$.
    \item All previous blocks on $C$ are also valid.
\end{itemize}
Throughput this paper, if we refer to a ``block'', then it is a valid block. In other words, we do not take a block that is not self-contained into account when we consider the blockchain, even if it has a valid block summary, or enough endorsements from the committee, or valid transaction sets, etc. We denote all blocks that are received by node $u$ till round $r$ by ${\cal B}_{val}(u,r)$.


\subsubsection{Honest blocks}
Nodes that are participating in the consensus will also mind the honesty of the block, i.e., whether they are sent by honest nodes. In normal situation, a block will be received by all honest nodes in that round. Hence, a block is suspicious if
\begin{itemize}
    \item there are different blocks or block summaries generated by the same leader of that round;
    \item it is received outside of their round.
\end{itemize}
We define the blocks that are not suspicious as honest blocks, i.e., the honest blocks are
\begin{itemize}
    \item received in the corresponding rounds;
    \item with no received blocks or block summaries that are from the same leader of the same round but are different.
\end{itemize}
We denote all honest blocks that are received by node $u$ till round $r$ by ${\cal B}_{hon}(u,r)$. Note that honest blocks only make sense in ``normal situations''. We will specify how nodes determine their current mode in Subsection~\ref{ss:mode}.

\subsection{Candidate Blocks}

At the start of each round, node $u$ runs the function ${\tt CandidateBlocks}(u,r)$ to determine the blocks to be considered for the canonical chain of this round. In normal mode, a node will only consider honest blocks for their chain selection and could confirm transaction using probabilistic metrics. In abnormal mode, a node will consider all blocks for their chain selection and be aware of the confirmed transactions might not be final.


\begin{algorithm}
\caption{Determine the candidate blocks ${\tt CandidateBlocks}({\cal B}_{hon}(u,r),{\cal B}_{val}(u,r),mode(u,r))$} \label{alg:candb}
\begin{algorithmic}
\State $r' \gets$ the last round that $mode (u,r')\ne {\tt 'normal'}$ 
\If{$mode(u,r)={\tt 'normal'}$}
    \State ${\cal B}(u,r) \gets \{B: B \in ({\cal B}_{hon}(u,r) \backslash {\cal B}_{hon}(u,r'+1)) \cup B_{val}(u,r') \wedge B \sim \textrm{fn}(u,r-1)$
\Else
    \State ${\cal B}(u,r) \gets \{B: B \in {\cal B}_{val}(u,r), B \sim \textrm{fn}(u,r-1)$
\EndIf
\Return ${\cal B}(u,r)$
\end{algorithmic}
\end{algorithm}
Here, $B \sim {\cal B}$ is defined as $B$ is not conflict with any block in ${\cal B}$, and $\textrm{fn}(u,r-1)$ is the block that is considered as final by node $u$ in round $r-1$, which we will explain later.


\subsection{Canonical Chain}

We use the following rules to determine the canonical chain:
\begin{enumerate}
    \item $C$ includes $\textrm{fn}(u,r)$, which is the newest finalized block considered by $u$ at round $r$.
    \item For two chains $C_1, C_2 \in {\cal B}(u,r)$ satisfying the first rule, select the longer chain.
    \item If the length of the chain is tied, select the chain with a newer block.
    \item If there are still multiple chains left, select the one that is received first.  
\end{enumerate}

The function ${\tt CanonicalChain}({\cal B}(u,r))$ is an implementation of these chain selection rules. 

\subsection{Finality}\label{ss:final}

A finality vector $F(u,r)$ is a vector including five values, each value is a hash pointer to a block. Then, if $F(u,r)=[H(B_0),H(B_1),H(B_2),H(B_3),H(B_4)]$ we define $\textrm{nv}(u,r)= B_0$, $\textrm{pp}(u,r)=B_1$, $\textrm{pc}(u,r)=B_2$, $\textrm{cm}(u,r)=B_3$, and $\textrm{fn}(u,r)=B_4$, representing the blocks that are seen as the provisioned view, being prepared, being pre-committed, being committed, and finalized by node $u$ in round $r$, respectively.
A pruned finality vector $F_p(u,r)$ excludes the last value of $F(u,r)$.
Then, we denote such a pruned finality vector included in block $B$ by $F(B)=[H(\textrm{nv}(B)),H(\textrm{pp}(B)),H(\textrm{pc}(B)),H(\textrm{cm}(B))]$. In the context of HotStuff BFT \cite{hotstuff}, $F(B)$ can be seen as the new view requests including a proposal of block $\textrm{nv}(B)$, prepare messages for block $\textrm{pp}(B)$, pre-commit messages for block $\textrm{pc}(B)$, and commit messages for block $\textrm{cm}(B)$ sent by the leader and all signed committee members in block $B$. 
Further, we use the notation $B \preceq B'$ if block $B$ is in an earlier round than or the same round as $B'$ and $B \succeq B'$ if block $B$ is in a later round than or the same round as $B'$.
We use the notation $u \vdash \textrm{pc}(u,r)=B$ to represent that $u$ sends a pc message for $B$ at round $r$. Note that this is a sufficient but not necessary condition for $\textrm{pc}(u,r)=B$.


We define a counter~$N_{\textrm{tp}}(v,b,C)$ to count the number of messages of one of the four above-mentioned types. It counts the number of distinct leaders and committee members that produce blocks~$\{B: B\in C \,\wedge\, \textrm{nv}(B)=v\,\wedge\, \textrm{tp}(B)=b\}$ where $\textrm{tp}=\{\textrm{nv},\textrm{pp},\textrm{pc},\textrm{cm}\}$ standing for ``new view'', ``prepare'', ``pre-commit'' and ``commit'', respectively. Similarly, we also define $N_{\textrm{tp}}(b,C)$ with the same definition with $N_{\textrm{tp}}(v,b,C)$ except that it ignores the new view check~$\textrm{nv}(B)=v$. If a block $B$ has $N_\textrm{nv}(B,C), C \in {\cal B}(u,r)$, we call $B$ a {\em view}, and we use the notation $V|_B[\mbox{conditions of }B]$ for a block $B$ that satisfies $N_\textrm{nv}(B,C)\geq 2f+1, C \in {\cal B}(u,r)$ as well as the given conditions and $\overrightarrow{V|_B}[\mbox{conditions of }B]$ for the latest such view.

Given the locally observed canonical chain, node~$u$ uses Algorithm~\ref{alg:final} to determine the finality vector~$F(u,r)$ that records the node's BFT consensus status. The vector is initialized as $F(u,0)=[H(b_0),H(b_0),H(b_0),H(b_0),H(b_1)]$ where $b_0$ is an unique identifier standing for ``null'' and $b_1$ the genesis block.


\begin{algorithm}
\caption{Updating the finality vector $F(u,r)$ based on ${\cal B}(u,r),C(u,r)$} \label{alg:final}
\begin{spacing}{1.2}
\begin{algorithmic}[1]
\State Let $B$ be the last block on $C(u,r)$.
\State $\textrm{type}(u,r)\gets\textrm{type}(u,r-1)$ where $\textrm{type}=\{\textrm{nv},\textrm{pp},\textrm{pc},\textrm{cm},\textrm{fn}\}$. We denote $C=C(u,r)$, $\mathcal{B}=\mathcal{B}(u,r)$, and $\textrm{type}(u)=\textrm{type}(u,r)$ to simplify the presentation below.
\State {\bf \# Update $\textrm{cm}(u)$.}
\If{$N_{\textrm{pc}}\big(\textrm{nv}(B),\textrm{pc}(B),C\big) \geq 2f+1$}
    \State $\textrm{cm}(u) \gets \textrm{pc}(B)$, $\textrm{pc}(u) \gets b_0$
\ElsIf{$N_{\textrm{cm}}\big(\textrm{cm}(B),{\cal B}(u,r)\big) \geq f+1,\,\textrm{cm}(B)\succ \textrm{cm}(u)$}
    \State $\textrm{cm}(u) \gets \textrm{cm}(B)$
\EndIf
\If{$\textrm{fn}(u)\prec \textrm{cm}(u)$}
    \State $\textrm{fn}(u) \gets \textrm{cm}(u)$
\EndIf
\State {\bf \# Update $\textrm{pc}(u)$.}
\If{$\exists B: \Big(\exists B_1' \in C' \subseteq {\cal B} : N_\textrm{pp}(B_1', B, C') \geq 2f+1 \wedge  N_\textrm{pc}(B_1', B', C') = 0, \forall B' \perp B) \wedge ( \forall C'' \subseteq {\cal B}, V|_{B_2''}[B_2'' \succ B_1', B_2'' \in C''] : N_\textrm{pc}(B_2'', B, C'') > 0 ) $}
\State{$B_\textrm{RTPC} \gets B$}
\EndIf
\State{$\hat{B} \gets \textrm{pc}(u,\max_{\textrm{pc}(u,r') \ne b_0}(r'))$}
\If{$\exists r_1': u \vdash \textrm{pc}(u,r_1')=\hat{B}$}
    \State{$\hat{B}_1 \gets \textrm{nv}(u,\max_{u \vdash \textrm{pc}(u,r_1')=\hat{B}}(r_1'))$}
\EndIf
\If{$B_\textrm{RTPC}$ exists $\wedge (\hat{B} \sim B_\textrm{RTPC} \vee \hat{B} = b_0 \vee \hat{B}_1$ does not exist$)$}
    \State{$\textrm{pc}(u) \gets B_\textrm{RTPC}$}
\ElsIf{$B_\textrm{RTPC}$ exists $\wedge \exists V|_{B_2'}[B_2' \succ \hat{B}_1, B_2' \in C_2', N_\textrm{pc}(B_2',\hat{B},C_2')=0]$}
    \State{$\textrm{pc}(u) \gets B_\textrm{RTPC}$}
\EndIf

\State {\bf \# Rule for unlocking $\textrm{pc}(u)$.}
\If{$B_\textrm{RTPC}$ exists $\wedge \textrm{pc}(u) \ne B_\textrm{RTPC} $}
\State{$\textrm{pc}(u) \gets b_0$}

\EndIf

\State {\bf \# Update $\textrm{pp}(u)$.}
\If{$\Big(N_{\textrm{nv}}(\textrm{nv}(B),C) \geq 2f+1\Big)  \wedge \Big(N_{\textrm{pc}}(\textrm{nv}(B),B',C)=0 \,\,\textrm{for}\,\, \forall B' \in {\cal B} \wedge B' \perp \textrm{nv}(B)\Big) $}
            \State $\textrm{pp}(u) \gets \textrm{nv}(B), \,\textrm{pc}(u) \gets b_0$
\EndIf

\State {\bf \# Update ${\textrm{nv}}(u)$.}
\If{$\textrm{nv}(B) \succ \textrm{nv}(u)$}
    \State $\textrm{nv}(u) \gets \textrm{nv}(B)$
\ElsIf{$\Big(\textrm{nv}(B) \perp \textrm{nv}(u)\Big) \vee \Big( \textrm{nv}(u) = b_0\Big)$}
    \State $\textrm{nv}(u) \gets B$ 
\EndIf
\If{$N_{\textrm{nv}}(\textrm{nv}(B),C) \geq 2f+1$}
    \State $\textrm{nv}(u)\gets B$
\EndIf
\State {\bf \# Unlock the prepare block if it conflicts the provisioned view block.}
\If{$\textrm{pp}(u) \perp \textrm{nv}(u)$}
    \State $\textrm{pp}(u) \gets b_0$
\EndIf

\State \Return{$H(\textrm{nv}(u)),H(\textrm{pp}(u)),H(\textrm{pc}(u)),H(\textrm{cm}(u)),H(\textrm{fn}(u))$}

\end{algorithmic}
\end{spacing}
\end{algorithm}

Let $C$ be the canonical chain of round $r$, here, we describe the rules that node $u$ uses in round $r$ in Algorithm~\ref{alg:final} in words:

\begin{itemize}
    \item {\bf Commit rules:}
    \begin{enumerate}
        \item $u$ commits a block $B$ if he has received $2f+1$ pc messages of $B$ in one view.
        \item $u$ commits a block $B$ if he receives $f+1$ cm messages for $B$.
    \end{enumerate}
    \item {\bf Pre-commit rules:} There exists a block that is ``ready to pre-commit'':
    \begin{enumerate}
        \item There exists a view $B_1' \in C' \subseteq {\cal B}$ such that $N_\textrm{pp}(B_1', B, C') \geq 2f+1$ and $N_\textrm{pc}(B_1', B', C') = 0$ for all $B' \perp B$;
        \item For all $C'' \subseteq {\cal B}$ and $V|_{B_2''}[B_2'' \succ B_1', B_2'' \in C'']$, it holds that $N_\textrm{pc}(B_2'', B, C'') > 0$.
    \end{enumerate}
    
    Then, let us denote the latest pre-committed block of $u$ by $B'$, i.e., $B' = \textrm{pc}(u,\max_{\textrm{pc}(u,r')\ne b_0}(r'))$. 
    We denote the latest view that $u$ sends a pc message for $B'$ by $B_1'$, i.e., $B_1'= \textrm{nv}(u,r_1)$ where $r_1 = \max_{u \vdash \textrm{pc}(u,r_1')=B'}(r_1')$. Now, $u$ pre-commits $B$ in one of the two cases:
    \begin{enumerate}
        \item $B' \sim B$ or $B'$ does not exist or $B_1'$ does not exists;
        \item $B' \perp B$, but there exists a view $V|_{B_2'}[B_2' \succ B_1', B_2' \in C_2', N_\textrm{pc}(B_2',B',C_2')=0]$.
    \end{enumerate}
            \item {\bf Unlock the pre-committing block:} If $\textrm{pc}(u,r) = B$, he set $\textrm{pc}(u,r) = b_0$ if $B$ is no longer ``ready to pre-commit'' according to the current ${\cal B}(u,r)$.
    \item {\bf Prepare rules:} $u$ prepares $B$ if $u$ receives $2f+1$ nv messages of $B$ with no pc message for a conflicting block $B' \perp B$ on the canonical chain.
    \item {\bf New view rules:}
    \begin{enumerate}
        \item giving a previous block $B$, if $\textrm{nv}(u,r-1) \prec \textrm{nv}(B)$, then $\textrm{nv}(u,r) \gets \textrm{nv}(B)$.
        \item giving a previous block $B$, if $\textrm{nv}(u,r-1) \succeq \textrm{nv}(B)$ and $\textrm{nv}(u,r-1) \perp \textrm{nv}(B)$, then $\textrm{nv}(u,r) \gets B$.
        \item set block $B$ as the provisioned new view if $2f+1$ consistent nv message has been received.
    \end{enumerate}
    \begin{itemize}
        \item {\bf Unlock the preparing block} (line 32) if the block that it prepares conflicts the provisioned new view block.
    \end{itemize}
\end{itemize}

\subsection{Role determination}
The leader and the committee members are selected by using a hash function and a VRF, respectively, on the random beacon of the epoch and the round number.

\subsubsection{Random beacon}

A random beacon of epoch $e$, denoted by $b_e$, is determined in the epoch $e-1$. We find the $\hat{\cal E}$ contained in the last finalized block before round $r_{e-1}-\tau_0-\tau_1$, where $r_{e_1}$ is the last round of epoch $e$, $\tau_0$ is a parameter set according to the estimated maximum latency of block propagation, and $\tau_1$ is an estimated maximum latency of finality. We then compute $b_e$ by $b_{e}=H(\hat{\cal E})$.

\subsubsection{Role determination function}
We then use the ${\tt RoleDetermine}$ function to determine the roles of nodes in round $r$, which determines the leader of each round by a hash function and determines the committee by comparing a random number generated by a VRF function to a predetermined threshold $\epsilon$, which is set such that the expected number of committee members is $c$. 

Here, we use the ECDSA-based VRF scheme proposed in \cite{vrf}, where the VRF could be abstracted as the following.
Node $u$ could use his private key $sk_u$ to compute a random number $\beta_u$ with an arbitrary input $\alpha$ by:
\begin{equation}
    \beta_u= f_{\textrm{VRF}}(\alpha, sk_u)
\end{equation}
Then, node $u$ could provide a proof:
\begin{equation}
    \pi_u=\Pi_{\textrm{VRF}}(\alpha ,sk_u).
\end{equation}
Any node could use the public key of $u$, denoted by $pk_u$ and $\Pi_u$ to verify that $\beta_u$ is {\bf collision free} defined similarly to a cryptographic hash function, {\bf pseudorandom} in the sense that it is indistinguishable from a random number created by another node, and {\bf unique} in the sense that each $\alpha$ corresponds to a unique $\beta_u$. Moreover, the ECDSA-based VRF scheme proposed in \cite{vrf}, we also have
\begin{equation}
    \beta_u=H(\pi_u).
\end{equation}
Further, we define a mapping function $M(x)= u$ that will map an arbitrary input $x$ to a node $u$ with uniform probability. Now we introduce function ${\tt RoleDetermine}(u,r,C(u,r))$.

\begin{algorithm}
\caption{Role Determination ${\tt RoleDetermine}(u,r,C(u,r))$} \label{alg:role}
\begin{algorithmic}
\State $e \gets$ the current epoch by $r$
\State $b_e \gets$ the random beacon of this epoch with $C(u,r)$
\If{$u == M(b_e|r)$}
    \State $role(r) \gets {\tt 'leader'}$
\EndIf
\If{$H(f_{\textrm{VRF}}(b_e|r,sk_u)) \leq \epsilon$}
    \State $role(u) \gets {\tt 'committee'}$
\EndIf
\Return $role(u)$
\end{algorithmic}
\end{algorithm}

\subsection{Committee member's procedure}

If node $u$ is a committee member in round $r$, in the time interval of $[\Delta, 2\Delta]$ in round $r$, he will wait for the block summary as well as the transaction set sent by the leader of the round.
Then, he calls the ${\tt ValBlockSum}(s',sig',{\cal TX},C(u,r))$ to validate the block summary. If the result is ${\tt 'valid'}$, it endorses this block summary and broadcasts its endorsement. The expected number of the committee members of each round is $c$ and the leader waits for the endorsements from $d<c$ committee members, combines them, and sign them as shown in Algorithm~\ref{alg:np}. Here, $c$ and $d$ are set accordingly to the probability of misbehavior in the normal situation such that with high probability, a block could be proposed in each round.

\subsubsection{Validation of the block summary}

The block summary is validated with ${\tt ValBlockSum}(s',Sig',{\cal TX},C(u,r))$.

\begin{algorithm}
\caption{Validate Block Summary ${\tt ValBlockSum}(s',Sig',{\cal TX},C(u,r))$} \label{alg:vbs}
\begin{algorithmic}
\State $s'$ is a summary in the form of $H(h(B'))|r'|MR({\cal TX})'|F'$;
\State $B \gets$ the last block on $C(u,r)$;
\If{$B' == B \wedge r' == r \wedge F(u,r)==F'  \wedge Sig' \mbox{ is } s' \mbox{ signed} \mbox{ with the correct key of the leader of round }r$}
    \State \Return ${\tt 'valid'}$
\EndIf
\end{algorithmic}
\end{algorithm}

\subsubsection{Endorsement}

If the result of the validation of the block summary is ${\tt 'valid'}$, node $u$ endorses it by broadcast $E_u(s')={\tt Endorsement}(s', u)$.

\begin{algorithm}
\caption{Endorsement algorithm ${\tt Endorsement}(s',u)$} \label{alg:eds}
\begin{algorithmic}
    \State $\pi_u \gets \Pi_u(b_e|r,sk_u)$
    \State \Return $(\pi_u, Sig_u(s'))$
\end{algorithmic}
\end{algorithm}

Here, $\pi_u$ should be broadcast with the endorsement so that the role of $u$ can be verified.

\subsection{Mode switch}\label{ss:mode}
In SCEM, the main function has two modes, normal mode and abnormal mode, to deal with normal and abnormal situations in the network, respectively. In the normal situation, honest nodes will not consider dishonest blocks, except for the dishonest block they received before this period of the normal situation. In the abnormal situation, all blocks will be considered. Moreover, honest node will not confirm any blocks in abnormal mode.

In fact, honest nodes are encouraged to use any on-chain and off-chain information to determine their modes as long as it has a low false negative probability for abnormal situations, i.e., if an honest node confirms a block, then the probability that the node is actually in abnormal situation should be smaller than a given $\epsilon_2$. In fact, if honest nodes are rational, they would naturally stop confirming transactions if they suspect the network is abnormal. Here, we propose a practical mode switching mechanism based on randomly pinging other nodes.

\begin{algorithm}
\caption{Determine the mode according to the network ${\tt Mode}(u,r)$} \label{alg:smode}
\begin{algorithmic}
\State pinging node with $H(\pi_u|v) \mod{n} \leq c$, where $v$ is the id of the node.
\If{more than $\frac{c(n-f)}{n}$ nodes response in time}
    \State $\textrm{cnx}(r)={\tt normal}$
\EndIf
\If{$\textrm{cnx}(r-k+1),\textrm{cnx}(r-k+2),\ldots,\textrm{cnx}(r)={\tt normal}$}
    \State $\textrm{mode}={\tt normal}$
\Else
    \State $\textrm{mode}={\tt abnormal}$
\EndIf
\Return $\textrm{mode}$
\end{algorithmic}
\end{algorithm}

\section{Mechanisms in SURFACE}\label{s:mec}

In this section, we explain the reason behind the designs introduced in Section~\ref{s:SURFACE} with context. We decompose the algorithm into several mechanisms and compare them to their counterparts in existing works, and clarify the similarities and differences, as well as our reason of choice.

\subsection{Round based leader selection}

With the practical asynchronous network assumption, we use a round based leader selection mechanism which has been widely used in existing works, in particular, proof-of-stake (POS) algorithms like \cite{snowwhite,ouroboros,algorand,dfinity}. These algorithms are different in three aspects:
\begin{itemize}
        \item {\bf Sole leader vs. Committee}: The single leader approach is a straightforward extension of Bitcoin. It is followed by POW based algorithms like \cite{ng} as well as POS based algorithms like \cite{snowwhite,ouroboros,praos}. These algorithms achieve Nakamote-like consensus that require several rounds to probabilistically confirm with no finality. In POW based algorithms like \cite{hybrid,byzcoin} and POS based algorithms like \cite{algorand,dfinity}, a random selected committee is used to achieve immediate finality with negligible fault probability. However, it introduces a higher message complexity due to the communication in the committee. Moreover, the fault tolerance ratio (the ratio of allowed adversaries in the total population) should be calculated carefully in order to guarantee that less than 1/3 of the committee are adversarial.
    \item {\bf The choice of random function}: Hash function is one of the most straightforward choice as a random oracle, which is used by \cite{snowwhite,ouroboros}. However, it suffers the disadvantage of predictability, which can be exploit by the adversaries. Then, VRF is used by \cite{algorand,praos} so that the role of a node is not known until himself revealing it with a proof. However, the number of leader or committee members is not definitive. In \cite{dfinity}, the BLS threshold signature is used for it is both deterministic and unpredictable.
    \item {\bf Random beacon}: A random beacon is required to generate a pseudo-random number, which must have reached consensus and could not be manipulated by the adversaries. An epoch based random beacon is used by most algorithms, where an epoch is a period that is sufficiently long for nodes to reach consensus on the beacon. Then, in order to prevent manipulation, \cite{snowwhite} uses the concatenation of randomness from many nodes so that adversaries could not manipulate all of them and in \cite{ouroboros,algorand}, the randomnesses are proposed with a commitment scheme to further prevent manipulation. Then, the beacon is used for an epoch so that it is impossible to successfully bias the selection of the whole period.
\end{itemize}

In SURFACE, we aim at a consensus algorithm that ``almost fork-free'' under normal situation in practical so that the bandwidth will not be wasted on transmitting blocks that are eventually discarded.
Hence, we choose a leader plus committee approach, which is similar to \cite{algorand}. The main difference in here is that in \cite{algorand} the size of the committee should be chosen sufficiently large so that the super majority (more than 2/3 of the population) of the committee are honest. However, in SURFACE, this is not a hard requirement as we only aim for ``almost fork-free'' but not absolute fork free. As a result, SURFACE could be seen as a generalization of leader based approach and committee based approach. It can achieve fast confirmation or even immediate finality if the committee size is large, and has a less message complexity but longer confirmation time if the committee size is small.

Second, for the random function, we use the straightforward hash mapping to select the leader and uses VRF to select the committee. The VRF committee selection will limit the capability of adversarial leaders to collude with the committee. Then, the hash mapping will guarantee that each round having exactly one leader, which will not result in empty rounds or inconsistency caused by multiple leaders in one round like in \cite{praos}. However, it does result in predictability of the leaders of the next epoch, which gives a chance of corruption attack. In our assumption, we assume the ``one-epoch ahead'' predictability is acceptable.

Third, we also uses the randomness from many nodes of the previous epoch to determine the random beacon. In particular, the random beacon is determined by the randomness created by VRF in the last finalized and received block in the previous epoch. It is guaranteed to be consistent for all nodes as it is finalized. Then, it could be manipulated with a non-negligible probability, which is the main different of SURFACE and \cite{algorand}. However, we will later show that we could still achieve probabilistic consistency even if the adversaries could manipulate the random seed.

\subsection{Optimizing throughput}
In SURFACE, many schemes are used to achieve an optimized throughput in practical use cases, especially in a industrial oriented consortium blockchain where the nodes are considered trusted in a certain degree.

\subsubsection{One round, one block}
The idea of ``one round, one block'' is the most straightforward approach of making a ``blockchain''. However, if we consider reaching consensus on messages rather than ``blocks'', then we have alternations like BFT algorithms and directed acyclic graphs (DAGs).

Classical BFT algorithms \cite{bracha,benor,castro} has $O(n^2)$ message complexity per consensus, where $n$ is the number of nodes in the network. Some more recent BFT algorithms could reduce the message complexity to $O(n)$ \cite{zyzzyva,byzcoin,700bft,miller,hotstuff}. The finality mechanism of SURFACE is inspired by HotStuff BFT \cite{hotstuff}, which also organizes the consensus into a chain of blocks. We will later clarify the similarities and differences between our algorithm and HotStuff BFT.

DAG based consensus algorithms allow multiple nodes to propose blocks simultaneously and eventually reach consensus on a graph instead of a chain \cite{tangle,spectre,phantom,conflux,avalanche,prism} \footnote{Note that these algorithms are different in many aspects with the only similarity of organizing data in the form of DAG.}. However, these algorithms are in general more complicated, especially to order transactions, and there are no clear evidence that they achieve higher throughput, lower latency, and/or have better bandwidth efficiency than the chain based approaches.

Another alternation is scale-out consensus algorithms like sharding approaches \cite{omniledger,chainspace,ren2} or off-chain approaches like \cite{lightning,plasma}. However, all these algorithms compromise in security, decentralization, or functionality, which are not suitable for our case.

Hence, we use a one-round one-block approach for our blockchain, which is also in line with many state-of-the-art consensus algorithms like \cite{ouroboros,algorand,grandpa,dfinity}.

\subsubsection{Decomposition of the blocks}
In our algorithm, the block is decomposed into multiple parts and the actual block is never broadcast together. The purpose of this design is to reduce communication redundancy. In particular, the transactions will only be broadcast once by the leader, instead of twice, i.e., first sent to the committee for validation, then broadcast with the complete block to the whole network.

There are other works that partially address this problem. In \cite{algorand,oguzhan}, the redundancy of the ``the transactions is broadcast twice, once before the block and once in the block'' is mitigated by only sending the hash indices of the transactions in the block.
\subsubsection{Delayed validation}

In our algorithm, the transaction set included in the block of round $r$ is actually not validated by the committee members of round $r$. It is validated by the leader and the committee members of round $r+1$ to decide whether to append blocks to it.



The reason behind this choice is to reduce the latency and wasted bandwidth caused by the validation of the transactions. If the committee validates the transaction of this round, then, the leader has to wait for the committee to validate the transactions and response with the endorsements, during which the bandwidth is wasted. The delayed validation design will allow the leader to fully utilize the bandwidth of a round for block transmission, while the validation could be done by the committee in the next round, while they are receiving the block of the next round. This is essentially optimistically trading the bandwidth wasted on waiting for the validation results for the possible bandwidth wasted on receiving invalid blocks, which will improve throughput in general as we assume that the network is mostly in normal situation and it is almost fork-free.

Similar approaches have also appeared in \cite{ncmax,grandpa}, however, in different forms and for different purposes.

\subsection{Finality}
To achieve finality, we include a finality vector in each block, which could be seen as a collection of new view messages, prepare messages, pre-commit messages, and commit messages in HotStuff BFT \cite{hotstuff}. In other words, a finality vector $F(B)$ included by a valid block $B$ is a collection of these consensus messages send by the leader and all endorsed committee members of block $B$. Then, as the chain grows, these messages will be collected to proceed a three-phase commit approach, i.e., $B$ will be proposed in a new view when $2f+1$ new view messages are received, be prepared when $2f+1$ prepare messages are received, and be pre-committed when $2f+1$ pre-commit messages are received, as shown in Figure~\ref{fig:normal}.
We use a view-change mechanism that is similar to HotStuff BFT, where nodes send pre-commit messages along with the new view messages, and a new is entered and a new block is only seen as proposed if it receives $2f+1$ new view messages without pre-commit messages for a conflicting block, which is an indication of no conflicting block that has been committed. On the other hand, if there exists a conflicting committed block $B$, there is at least one honest node that will send pre-commit message of $B$ along with the new view message and all nodes will eventually receive the chain of $B$ and pre-commit $B$.
This mechanism guarantees both liveness and consistency.

\begin{figure}
\centering
\includegraphics[width=0.95\textwidth]{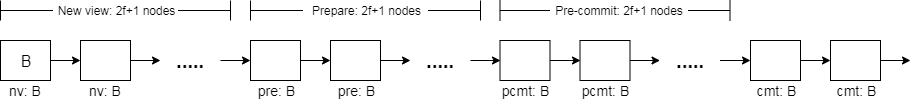}
\caption{Normal procedure of finality in SURFACE}
\label{fig:normal}
\end{figure}

However, although the finality mechanism in SURFACE uses a similar consensus process as HotStuff BFT, it is very different from HotStuff BFT by its nature. Firstly, as SURFACE is a chain based algorithm that uses a BFT based mechanism to achieve finality, it is crucial that the finality is chain compliant:
\begin{definition}[Chain compliant finality]
If a block $B$ is finalized by an honest node $u$, then another honest node $v$ could also finalize $B$ even if he has only received the chain $C(:B)$.
\end{definition}
Chain compliant finality suggests that all consensus messages required to commit a block should be included on the chain that it is committed. As a result, the consensus messages are subjective to their chains, which is not the case for normal BFT consensus algorithms.

The second difference is that honest nodes should follow canonical chain rules to only select and extend the heaviest chain if they have locally observed multiple candidate chains that are not finalized. Hence, honest nodes could vote in multiple QCs due to forks, which will lead to confusion. To address this problem, we let nodes always change their new view to a newer block as shown in Figure~\ref{fig:changeview} so that honest nodes will not simultaneously vote for multiple QC.

\begin{figure}
\centering
\includegraphics[width=0.95\textwidth]{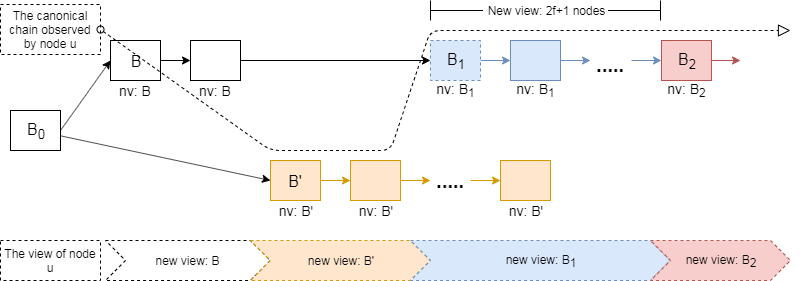}
\caption{Nodes will always provision a newer view on the canonical chain.}
\label{fig:changeview}
\end{figure}

In rare situations, a block with a heavier weight, denoted by $B'$, is not the block that is the closest to reach finality, denoted by $B$. In this case, the consensus process will continue on the chain of $B'$ until $B$ is committed, as shown in Figure~\ref{fig:pcmtchangeview}. Then, the chain of $B'$ will be discarded and nodes will resend the pre-commit messages of $B$ in the chain of $B$ for the sake of chain compliant finality. 

\begin{figure}
\centering
\includegraphics[width=0.95\textwidth]{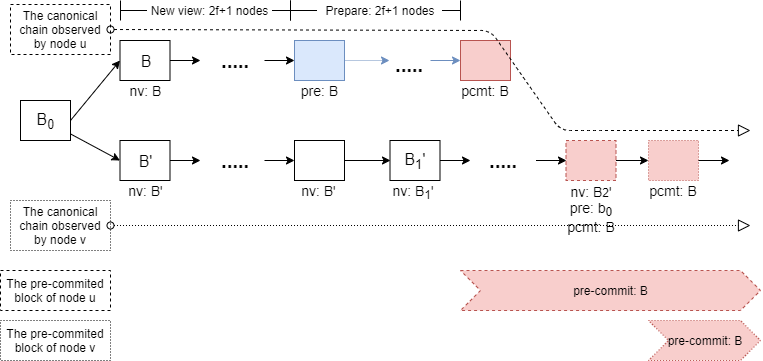}
\caption{Node $u$ switches to a chain when he has already prepared and pre-committed a conflicting block. Note that in here, node $u$ and $v$ will still considers the chain of $B'$ is the canonical chain until they commit $B$. Then, they will discard the chain of $B'$ and change to $B$.}
\label{fig:pcmtchangeview}
\end{figure}

In HotStuff BFT, if nodes are locked on different views, they will wait for the leader to send a high QC (the QC in the highest view), then unlock their current view and change their views according to the high QC. However, in SURFACE, since a message can only be sent if the leader and the committee both agree with it, the liveness will not be guaranteed if nodes prepare or pre-commit for many different blocks and only unlock their current view when $2f+1$ new messages are received. Hence, we let nodes first unlock their current view as soon as they discover another chain preparing a newer block and realize that their views are not the newest, then change to the view with the high QC when $2f+1$ messages are received. In Figure~\ref{fig:newpre}, we show that a node prepared for a block $B$ will unlock if he provisions to a conflicting view.
\begin{figure}
\centering
\includegraphics[width=0.95\textwidth]{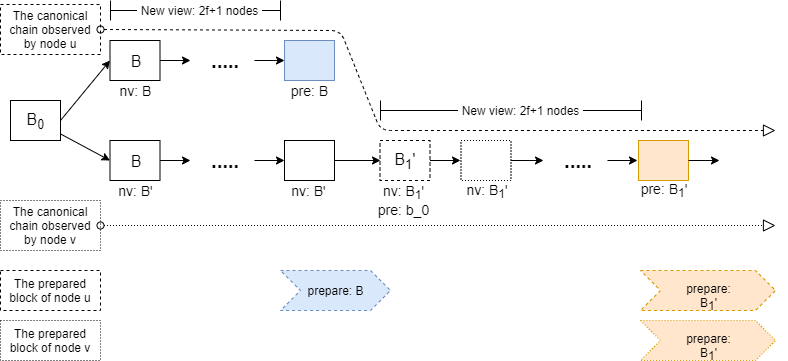}
\caption{Node $u$ unlocks its prepared block.}
\label{fig:newpre}
\end{figure}

\section{Performance Analysis}\label{s:pa}



SURFACE could achieve faster confirmation by letting more nodes endorse a block. We prove that a fork could only be created by an adversarial leader colluding with an adversarial committee in the normal situation. Here, an adversarial committee suggests that there are at least $d$ adverserial committee members. As a result, the blocks in SURFACE are confirmed faster comparing to the scenario that a sole leader is used.

\subsection{Forks}

A fork is defined as two chains $C_1$ and $C_2 \perp C_1$ which are returned from ${\tt Canonical Chain}({\cal B}(u_1,r),F(u_1,r-1))$ and ${\tt Canonical Chain}({\cal B}(u_2,r),F(u_2,r-1))$, where $u_1$ and $u_2$ are two honest nodes. Here, $C_2 \perp C_1$ is defined as there exists a block set ${\cal B}_1 \subset C_1, {\cal B}_1 \cap C_2 = \emptyset$ and a ${\cal B}_2 \subset C_2, {\cal B}_2 \cap C_1 = \emptyset$. 
We concern about the probability of existing a fork with depth $k$, where $k$ is a parameter for confirmation. Clearly, a fork of depth $k$ is able to be exploit to perform a double spending attack.

\subsection{Forks with colliding blocks}
Firstly, we consider the case in which there exists colliding blocks, i.e., there are two blocks that $B_1 \in {\cal B}_1$ and $B_2\in{\cal B}_2$ of the same round.

In this case, as the leader and committee are deterministic and fixed for each round given that the random beacons are consistent, the leader of round $r$ must be adversarial to create $B_1$ and $B_2$. Then, in order to create inconsistency, he sends these two blocks, in particular, the last part of the block, in the very end of round $r$. However, by our endorsement rule, the committee members will only endorse for a block summary after a period of $\Delta$ in round $r+1$. Then, by the gossip communication model in normal situation, both blocks will be received and be considered as suspicious by honest nodes and discarded. Hence, a $k$-depth fork with colliding block must be caused by $k$ consecutive adversarial leaders and committee members in synchronous scenario.

\subsection{Forks without colliding blocks}
There is another scenario in which the fork could be created and extended without colliding blocks. Let consider the following case:

At round $r+1$, the adversarial leader $l_{r+1}$ creates a block $B_{r+1}$ appending to the block in round $r$, denoted by $B_r$. He follows the normal procedure to broadcast the block summary, collect the endorsements, except for that he ``fraudulently delays'' the broadcast of the collected endorsements. As a result, $B_{r+1}$ is not actually broadcast, and thus not acquired by the rest of the network. Then, the adversarial leader of round $r+2$, $l_{r+2}$, creates a block appending to $B_r$. It will be endorsed by the honest committee as $B_r$ is the latest block that they observed. However, $l_{r+2}$ again ``fraudulently delays'' the broadcast of $B_{r+2}$. At the meantime, $l_{r+1}$ broadcasts $B_{r+1}$. Then, the adversarial leader of round $r+3$ will perform the same strategy to extend the chain of $B_{r+1}$. Further, the leader of round $r+4$ could also use the same strategy to extend $B_{r+2}$. Note that although the chain of $B_{r+2}$ has the same length as the chain of $B_{r+1}$, however, by our chain selection rule, $B_{r+2}$ is the latest and should be selected. So on and so forth, adversaries could create a $k$-depth fork with $2k$ adversarial leaders with honest committees.

However, this type of forks are addressed by our rules of honest blocks. In normal mode, the blocks which are received outside of its round will be considered as suspicious and will not be taken into account for canonical chain selection.


\subsection{False confirmation in abnormal situation}
In the abnormal situation, a fork with colliding blocks could be created with less than $d$ malicious committee members. More precisely, since a malicious leader could propose colliding blocks and if the network is partitioned and the messages between committee members are delayed, honest committee members could endorse for a block $B$ while the other honest committee members endorsing for a block $B'$. Then, as the committee size is not fixed, a fork could be created by a malicious leader colluding with $m<d$ malicious committee members. Additionally, there are $2(d-m)$ honest committee members are parted into two groups and all intermediate messages between these two groups are delayed.

However, as committees are selected by VRF and the network is not hostile by assumption, the adversaries could not predict who are the honest committee members and could only perform this type of attack by dividing the network before sending colliding blocks.
This attack will be countered by our mode switching mechanism in which honest nodes will detect a partition of the network with a same probability of a malicious committee. As a result, the probability of false confirmation decays at a same rate as the fork rate in the normal situation.



\subsection{Numerical analysis}\label{ss:na}

As discussed previously, a violation in consistency occurs with a $k$-depth fork, which requires at least $k$ consecutive malicious leaders and committees in the normal situation.
Then, in asynchronous but not hostile networks, it occurs when adversaries conduct a network partition attack and successfully created a fork by less than $d$ malicious nodes. At the meantime, the mode switching mechanism of two honest nodes in both parts of the network simultaneously failed to detect a network partition for $k$ consecutive rounds.

Moreover, we need to consider the possibility of biasness, i.e., if the beacon happens to be created by a malicious leader, then, he could exhaust approximately $b=\binom{c}{d}$ times to try to create a scenario of $k$ consecutive adversarial leaders and committees.


Then, we have 
\begin{align}
&   s_1 = \Big(\textrm{BC}(d,f,\frac{c}{n})\frac{f}{n}\Big)^k\binom{c}{d} \\
& \nonumber   s_2 = \textrm{BC}(\lceil \frac{c(n-f)}{n} \rceil , c ,\frac{n+f}{2n})^{2k}\binom{c}{d} \\
&\;\;\Big[ \sum_{d'=1}^{d} \Big( \textrm{BP}(d',f',\frac{c}{n}) \textrm{BC}(d-d',\lfloor \frac{n-d'}{2} \rfloor, \frac{c}{n})^2 \Big)\frac{f}{n} \Big]^k ,
\end{align}
where $s_1$ and $s_2$ are the probability of $k$-length fork occurs in normal and abnormal situation, respectively, taken into account the factor of biasness.
Here, the probability density function and the cumulative distribution function of binomial distribution of winning (at least) $x$ times in $n$ trials with probability $p$ are denoted by:
\begin{eqnarray}
    \textrm{BP}(x,n,p) &=& \binom{n}{x} p^{x}(1-p)^{n-x}, \\
    \textrm{BC}(x,n,p) &=& \sum_{x'=x}^{\infty} \binom{n}{x'} p^{x'}(1-p)^{n-x'},
\end{eqnarray}
respectively. In Table~\ref{tb:sp}, we give the security parameters in various network configurations.


\begin{table}
\centering
\begin{tabular}{|l|l|l|l|l|l|l|}
\hline
& Biasness & $s_1$ & $s_2$  \\ \hline
\begin{tabular}[c]{@{}l@{}}$c=10,d=7$\\ $k=7,f=33$\end{tabular} & 120  & $7.57\times 10^{-12}$  & $1.57\times 10^{-7}$ \\ \hline
\begin{tabular}[c]{@{}l@{}}$c=8,d=5$\\ $k=5,f=25$\end{tabular} & 56  & $8.14\times 10^{-9}$  & $4.12\times 10^{-12}$ \\ \hline
\begin{tabular}[c]{@{}l@{}}$c=6,d=4$\\ $k=4,f=20$\end{tabular} & 15  & $1.43\times 10^{-8}$  & $4.85\times 10^{-9}$ \\ \hline
\end{tabular}
\caption{The security parameters given in a network with $n=101$.}\label{tb:sp}
\end{table}

\begin{example}
We consider a network with 101 nodes, 33 adversaries, and a committee size of 10 with a requirement of 7 endorsements. Then, if an honest node sees a block $B$ appended by 7 blocks and he is in normal mode, he can confirm by knowing that the probability that there exists another honest node that confirms a conflicting block is smaller than $1.57 \times 10^{-7}$, regardless of whether the network situation is normal or not. If the block interval is 10 seconds, the confirmation time is 70 seconds, guaranteeing that an attack on the consistency could only happen approximately once per 2 years.
For comparison, to achieve the same level of security, a classical approach will requires $k=15$.

\end{example}

\subsection{Finality}

Besides a generalization of the Nakamoto-like consensus, SURFACE is also a generalization of BFT algorithms like \cite{algorand,hotstuff} with a flexible committee size. As a result, SURFACE could also achieve finality, i.e., the consistency condition in BFT consensus, in a similar fashion as \cite{hotstuff}.

\begin{theorem}[Finality]\label{th:cons}
In SURFACE, if an honest node $u$ has $\textrm{cm}(u,r)=B$, then, there could not be another honest node $u'$ that considers a conflicting chain is final, i.e., there cannot be another node $u'$ and a round $r'$ that has $\textrm{cm}(u',r')=B'$ and $B \perp B'$.
\end{theorem}

\begin{proof}
Assume that there exists two honest nodes $u$ and $u'$ that have $\textrm{cm}(u,r)=B$ and $\textrm{cm}(u',r')=B', B \perp B'$. By the commit rule, there must exist a chain $C$ that has $2f+1$ nv messages and $2f+1$ pp messages. Moreover, we assume there are $2f+1$ pc messages for $B$ in a view $B_1$, i.e., coming along with the nv messages for $B_1$. Similarly, there must exist a chain $C'$ that has $2f+1$ pc messages for $B'$ in a view $B_1'$. Hence, there must be a node $v$ that sends a pc message for $B$ at view $B_1$ and sends a pc message for $B'$ at view $B_1'$. W.l.o.g. we assume that $B_1' \succ B_1$.

By the nv rules, $v$ must send pc messages for $B$ before sending pc message to $B'$ since he cannot change from view $B_1'$ to a older view $B_1$. Then, as he has send a pc message for $B$, he can only send the pc message for $B' \perp B$ according to the second pc rule.
Hence, there must exist a chain $C''$ (could be the same as $C'$) that has $2f+1$ nv messages for a block $B_1'' \succ B_1$ with no pc message for $B$. Then, node $v$ observes the view $B_1''$ at round $r_0$ and afterwards pre-commits $B$. 

Then, there must be a node $v'$ that sends both the pc message for $B$ in the view $B_1$ and the nv message of $B_1''$ without pc message for $B$. Here, as $B_1'' \succ B_1$, by the new view rule, $v'$ must first send a pc message for $B$ then sends a nv message at view $B_1''$ without pc message for $B$, which suggests that he has unlocked $B$.
Then, according to the unlock rule, there must exist a chain $C'''$ that has $2f+1$ nv messages for a block $B_1''' \succ B_1$ with no pc message for $B$. Moreover, $v'$ must have observed this and unlocked $B$ at a round $r_0' < r_0$ so that $v$ could make the later observation and unlocking.

As there are finite rounds between $B_1$ and $B_1'$, there must exists a node that sends pc message in the view $B_1$ but unlocks $B$ without meeting the condition of unlocking pre-committed blocks.


\begin{figure}
    \centering
    \includegraphics[width=0.95\textwidth]{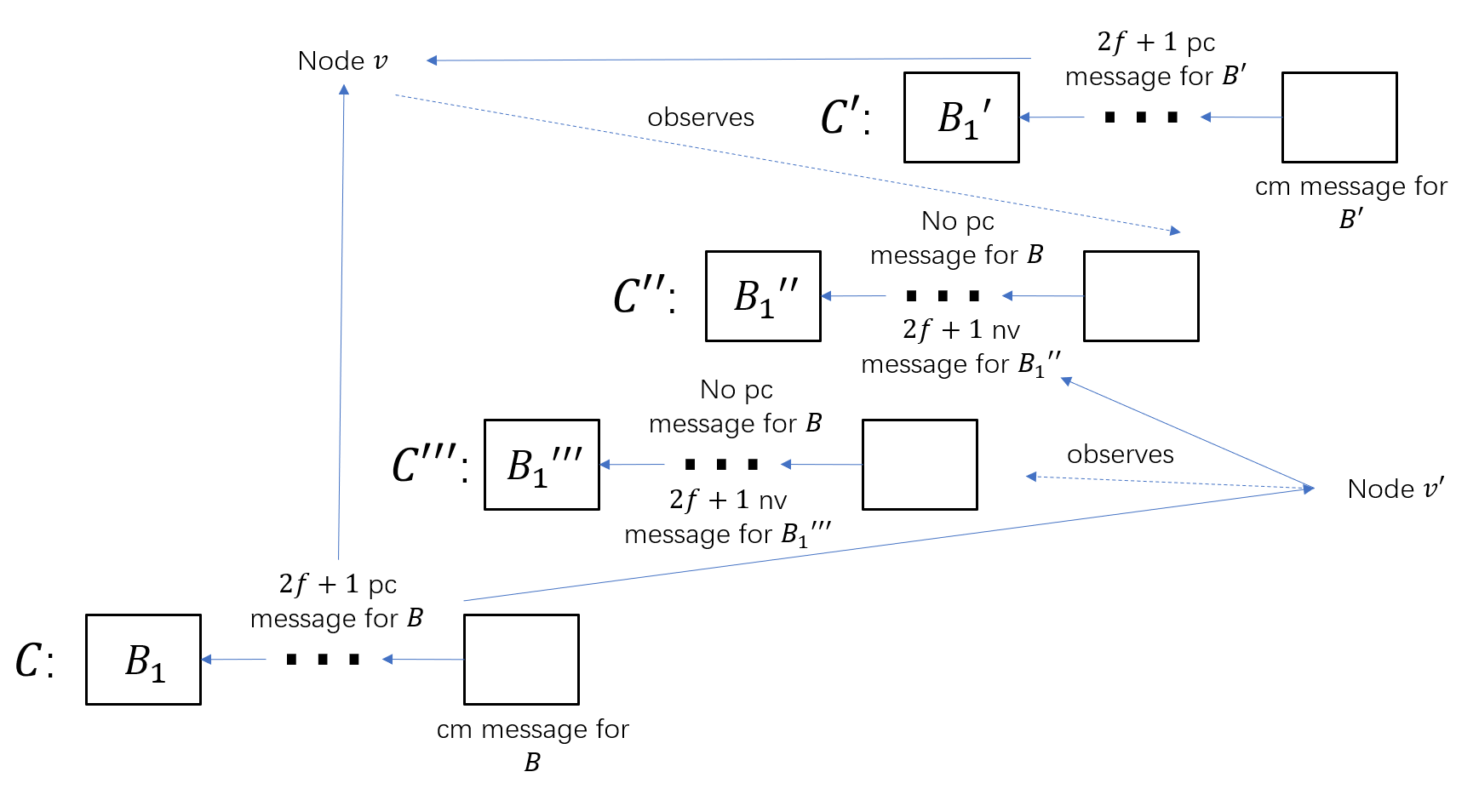}
    \caption{Illustration for proving Theorem~\ref{th:cons}}
    \label{fig:th_1}
\end{figure}

\end{proof}

\subsection{Synchronous liveness}

Firstly, we prove the following lemma showing that the system could enter a new view in ``normal'' situations.

\begin{lemma}\label{lm:chaingrow}
In the normal situation, new blocks can be proposed in a new view regardless of the finality vectors of the nodes.
\end{lemma}
\begin{proof}
By our assumption, it is straightforward that blocks will be proposed in a new view if the network is always in normal situation. We focus on the scenario that nodes start with different finality vectors, in particular, locked with different prepare and pre-commit blocks.
We show that there will eventually be large groups of nodes having consistent finality vectors if they have the same ${\cal B}(u,r)$ so that the leaders and committees could reach agreement and blocks could be proposed with non-zero probability.
Let us assume that honest nodes are pre-committing to a set of blocks ${\cal B}^*$ at round $r$. Moreover, by our assumption of the normal situation, all node should have a consistent ${\cal B}(u,r)$ and canonical chain $C$.


By the pre-committing rule, for each block $B_i \in {\cal B}^*$, there must exists a block $B_i'$ such that there exists a chain $C_i$ in which there are $2f+1$ nv messages for $B_i'$ coming along with $2f+1$ pp messages for $B_i$. Moreover, there are no pc messages for any $B_j \perp B_i$ coming along with these pp messages. We use the notation ${\cal B}'$ for the set for all $B_i'$, $B_0'$ for the newest block in ${\cal B}'$, $B_0$ for the block that receives $2f+1$ pp messages in the view $B_0'$, and $C_0$ for the chain that includes $B_0$, $B_0'$, and these $2f+1$ pp messages. 

Now, we discuss two cases: first, if there exists a view $B_0'' \succ B_0'$ with no pc message for $B_0$, then $B_0$ is not ``ready to pre-commit'' according to ${\cal B}(u,r)$ and will be unlocked by all nodes. Then, the process will continue with all nodes preparing and pre-committing new blocks. Second, if there dose not exist a view $B_0'' \succ B_0'$ with no pc message for $B_0$, the $B_0$ is ``ready to pre-commit'' according to ${\cal B}(u,r)$ by definition.

Then, by the pre-commit rule, all nodes that have send pc messages of $B_i$ before the view $B_0'$ will pre-commit $B_0$. On the other hand, all nodes that have send pc messages of $B_i$ after the view $B_0'$ will unlock their pre-committing blocks, but not yet pre-commit $B_0$. We denote the group of nodes that have $\textrm{pc}(u,r)=B_0$ by $\Omega$ and the group of nodes that have $\textrm{pc}(u,r)=b_0$ by $\Theta$.

Assume that there are $m$ honest nodes in group $\Omega$ and and $2f+1-m$ honest nodes in group $\Theta$. Then, in normal situation, with a probability of $\textrm{BC}(d,c,m+f)$ that a new block will be created by group $\Omega$ with a pre-commit message of $B_0$ and with a probability of $\textrm{BC}(d,c,n-m)$ that a new block will be created by group $\Theta$ with no pre-commit message. As a result, in the worst case, there is a probability of $2\textrm{BC}(d,c,n/2)$ that a block could be created. 
As the committee members of each block are independently and uniformly selected, nv messages of a view $B_2$ from $2f+1$ different nodes will be eventually received and all nodes will either unlock $B_0$ and prepare $B_2$ or all honest nodes will pre-commit $B_0$.
\end{proof}



With Lemma~\ref{lm:chaingrow}, we have the following theorem.

\begin{theorem}[Liveness]\label{lm:sl}
In normal network situation, new blocks could be committed.
\end{theorem}

\begin{proof}

By Lemma~\ref{lm:chaingrow}, all node will enter newer views. Then, in the normal situation, all node should be able to have a consistent canonical chain. Hence, by our algorithm, the proposed block would be able to collect $2f+1$ prepare messages, $2f+1$ pre-commit messages, and eventually be committed.





Note that leaders and colluding committees could try to propose blocks $B$ with a pre-commit message for arbitrary $B'$ in a view to delay the process. However, as we always assume that the probability of a colluding leader and committee is very low (very unlikely to happen during the consensus process) in the normal situation, we guarantee that new blocks will eventually be committed in the normal situation.






\end{proof}




\section{Comparison to other algorithms}\label{s:comp}

In general, one of the major difference between SURFACE and all other algorithms are the design of two modes, which distinguishes SURFACE from most blockchain consensus algorithms. The most similar ones are the Nakamoto consensus algorithms with finality, like Casper \cite{casper} and GRANDPA-BABE \cite{polkadots,grandpa}. Besides, we will also compare SURFACE to BFT algorithms like Algorand, HotStuff BFT, Tendermint, and Nakamoto consensus algorithms like Ouroboros.

\subsection{Casper and GRANDPA-BABE}
Casper \cite{casper} is a consensus algorithm that is designed to promote public blockchains like Ethereum with finality. Polkadots \cite{polkadots} uses a consensus scheme called GRANDPA-BABE \cite{grandpa}, which uses the finality arguments of Casper in a permissioned blockchain. Firstly, BABE is an algorithm that is similar to Ouroboros-Praos \cite{praos} that uses a VRF based approach to determine the block proposer of each round. Then, GRANDPA is used to allow nodes to spontaneously vote for the blocks and uses BFT arguments to reach finality.

As far as we know, GRANDPA-BABE is independently developed and is the most similar consensus algorithm to SURFACE in sense that it also uses a VRF based approach to guarantee an ever-growing chain and achieve finality upon that with BFT arguments. However, SURFACE and GRANDPA-BABE are different in the following aspects:

\begin{itemize}
    \item In a certain sense, the voting mechanism of GRANDPA is equivalent to a dynamic-size committee for each block. Hence, the performance of GRANDPA-BABE will differ from SURFACE depending on the network situation. In general, SURFACE will give a more stable confirmation time and a lower fork rate, while the performance of GRANDPA-BABE will depends on how motivated nodes are for voting. On the other hand, GRANDPA-BABE has not yet introduce an incentive scheme for the voters.
    \item GRANDPA does not have optimal responsiveness as stated in \cite{hotstuff}. In other words, in extreme scenario, honest nodes will have to wait for the maximum delay $\delta$ to make progress, which is not the case in SURFACE.
\end{itemize}

\subsection{HotStuff BFT}

The ``abnormal'' mode in SURFACE uses a similar but not identical approach as HotStuff BFT \cite{hotstuff} to achieve consensus. However, it is explained in Subsection~\ref{ss:final} that the finality mechanism used in SURFACE is different from HotStuff BFT in many aspects, e.g., chain compliant, unlocking pre-commit blocks.


Another difference is that SURFACE functions in the practical asynchronous network and HotStuff BFT functions in partial synchronous network, which is weaker assumption. However, although HotStuff BFT, as well as PBFT, could be straightforwardly use in the practical asynchronous network, the view changing time-out should be modified to a constant, otherwise they will suffer from a very long latency in the normal situation after a long period of abnormal situation.


The main advantage of SURFACE over HotStuff BFT is the random sample of committee allows a faster probabilistic confirmation speed in large networks in the ``normal'' mode. In HotStuff BFT, the block interval needs to be sufficient for $2f+1$ nodes to respond. However, in SURFACE, the block interval could be set smaller as the leader only needs to wait for the committee to respond.

\subsection{Algorand}
Algorand \cite{algorand} is similar to the ``normal'' mode of SURFACE in many aspects. However, Algorand achieves provable BFT with the leader and committee selected in each round, while SURFACE aims to only reduce the probability of forks. Note that there is a trade-off between the committee size and the fault tolerance in Algorand: in order to guarantee that the number of the adversaries in each committee is less than 1/3 by the law of the large number, the size of the committee and the ratio of adversaries in the network should be set accordingly. In SURFACE, this is not a concern as we allow the committee to be malicious and to create forks.

\subsection{Tendermint}

Tendermint \cite{tendermint} uses a two-phase BFT consensus approach, in which the process of committing a block is identical to PBFT. Then, in order to guarantee liveness, a node that has pre-committed for a block will unlock and prepare for a new block only if {\bf BOTH} of the two condition holds: 1) a block proposed in a newer view has received $2f+1$ prepare messages; 2) a certain timeout is reached. This approach is lack of responsiveness comparing to PBFT and HotStuff as suggested in \cite{hotstuff}, as a view change will only happen after the timeout, even if there are already $2f+1$ votes for a new view.

In our model, as we incorporate with the possibility of asynchronous and a random committee selection mechanism, a rather large timeout should be used to guarantee that there will be enough blocks to commit a block before timeout in case that nodes are locked on inconsistent blocks at the beginning of a consensus round. As a result, we choose to use a similar three-phase approach as HotStuff BFT to have responsiveness, i.e., the ability to make progress on the BFT process without needing to wait for any preset timeout period for view changes.




\section{Conclusion}\label{s:conc}

In this paper, we present SURFACE, a blockchain consensus algorithm that is especially designed and optimized for large real-world blockchains. The main reason behind the proposal of blockchain is the observation that in real-world, we tend to use double standards on the consensus algorithms used in blockchains. On one hand, it is commonly believed that synchronous consensus algorithms are not sufficient and suitable for blockchains. On the other hand, most blockchains function in highly synchronous networks and Bitcoin's POW actually has a very high requirement of synchrony. As a result, the proposed blockchain consensus algorithms are either theoretically sound for asynchronous case but not optimized in practice, or achieve sky-high performance in laboratory environments but vulnerable in extreme situations. Hence, we take both perspectives into account and put forth the practical asynchronous network model. We then propose SURFACE, which will give a near-optimal performance in the normal situation but still be able to reach definitive consistency in the extreme situations. Certainly, the drawback of SURFACE is that it will have sub-optimal performance if the network is different from our assumptions, e.g., the network shifts between multiple situations or partitions or attacks in network becomes a new ``normal'' situation for various reasons.
However, we believe SURFACE does fit the scenarios of most real-world blockchains and could provide a reasonably good performance for general cases.

\bibliographystyle{splncs04}
\bibliography{Implicit_Consensus}

\begin{thebibliography}{10}
\providecommand{\url}[1]{\texttt{#1}}
\providecommand{\urlprefix}{URL }
\providecommand{\doi}[1]{https://doi.org/#1}

\bibitem{chainspace}
Al{-}Bassam, M., Sonnino, A., Bano, S., Hrycyszyn, D., Danezis, G.: Chainspace:
  {A} sharded smart contracts platform. CoRR  \textbf{abs/1708.03778} (2017),
  \url{http://arxiv.org/abs/1708.03778}

\bibitem{redballoon}
Babaioff, M., Dobzinski, S., Oren, S., Zohar, A.: On bitcoin and red balloons.
  In: Proceedings of the 13th ACM conference on electronic commerce. pp.
  56--73. ACM (2012)

\bibitem{prism}
Bagaria, V., Kannan, S., Tse, D., Fanti, G., Viswanath, P.: Prism:
  Deconstructing the blockchain to approach physical limits. In: Proceedings of
  the 2019 ACM SIGSAC Conference on Computer and Communications Security. pp.
  585--602 (2019)

\bibitem{benor}
Ben-Or, M., Kelmer, B., Rabin, T.: Asynchronous secure computations with
  optimal resilience. In: Proceedings of the thirteenth annual ACM symposium on
  Principles of distributed computing. pp. 183--192. ACM (1994)

\bibitem{snowwhite}
Bentov, I., Pass, R., Shi, E.: Snow white: Provably secure proofs of stake.
  IACR Cryptology ePrint Archive  \textbf{2016}, ~919 (2016)

\bibitem{bracha}
Bracha, G.: Asynchronous byzantine agreement protocols. Information and
  Computation  \textbf{75}(2),  130--143 (1987)

\bibitem{casper}
Buterin, V., Griffith, V.: Casper the friendly finality gadget. arXiv preprint
  arXiv:1710.09437  (2017)

\bibitem{castro}
Castro, M., Liskov, B.: Practical byzantine fault tolerance. In: OSDI. vol.~99,
  pp. 173--186 (1999)

\bibitem{croman}
Croman, K., Decker, C., Eyal, I., Gencer, A.E., Juels, A., Kosba, A., Miller,
  A., Saxena, P., Shi, E., Sirer, E.G., et~al.: On scaling decentralized
  blockchains. In: International Conference on Financial Cryptography and Data
  Security. pp. 106--125. Springer (2016)

\bibitem{praos}
David, B., Ga{\v{z}}i, P., Kiayias, A., Russell, A.: Ouroboros praos: An
  adaptively-secure, semi-synchronous proof-of-stake blockchain. In: Annual
  International Conference on the Theory and Applications of Cryptographic
  Techniques. pp. 66--98. Springer (2018)

\bibitem{wattenhofer}
Decker, C., Wattenhofer, R.: Information propagation in the bitcoin network.
  In: IEEE P2P 2013 Proceedings. pp. 1--10. IEEE (2013)

\bibitem{eos}
EOS: \url{https://eos.io/}

\bibitem{oguzhan}
Ersoy, O., Ren, Z., Erkin, Z., Lagendijk, R.L.: Transaction propagation on
  permissionless blockchains: incentive and routing mechanisms. In: 2018 Crypto
  Valley Conference on Blockchain Technology (CVCBT). pp. 20--30. IEEE (2018)

\bibitem{minerdilemma}
Eyal, I.: The miner's dilemma. In: 2015 IEEE Symposium on Security and Privacy.
  pp. 89--103. IEEE (2015)

\bibitem{ng}
Eyal, I., Gencer, A.E., Sirer, E.G., Van~Renesse, R.: Bitcoin-{NG}: A scalable
  blockchain protocol. In: 13th USENIX Symposium on Networked Systems Design
  and Implementation (NSDI 16). pp. 45--59. USENIX Association (2016)

\bibitem{majority}
Eyal, I., Sirer, E.G.: Majority is not enough: Bitcoin mining is vulnerable.
  In: International Conference on Financial Cryptography and Data Security. pp.
  436--454. Springer (2014)

\bibitem{rational}
Ford, B., B{\"o}hme, R.: Rationality is self-defeating in permissionless
  systems. arXiv preprint arXiv:1910.08820  (2019)

\bibitem{grandpa}
Foundation, W..: Byzantine finality gadgets.
  \url{https://github.com/w3f/consensus/blob/master/pdf/grandpa.pdf} (2019)

\bibitem{garay}
Garay, J., Kiayias, A., Leonardos, N.: The Bitcoin Backbone Protocol: Analysis
  and Applications, pp. 281--310. Springer Berlin Heidelberg, Berlin,
  Heidelberg (2015). \doi{10.1007/978-3-662-46803-6-10}

\bibitem{algorand}
Gilad, Y., Hemo, R., Micali, S., Vlachos, G., Zeldovich, N.: Algorand: Scaling
  byzantine agreements for cryptocurrencies. In: Proceedings of the 26th
  Symposium on Operating Systems Principles. pp. 51--68. ACM (2017)

\bibitem{vrf}
Goldberg, S., Naor, M., Papadopoulos, D., Reyzin, L.: Nsec5 from elliptic
  curves: Provably preventing dnssec zone enumeration with shorter responses.
  IACR Cryptology ePrint Archive  \textbf{2016}, ~83 (2016)

\bibitem{700bft}
Guerraoui, R., Kne{\v{z}}evi{\'c}, N., Qu{\'e}ma, V., Vukoli{\'c}, M.: The next
  700 {BFT} protocols. In: Proceedings of the 5th European conference on
  Computer systems. pp. 363--376. ACM (2010)

\bibitem{dfinity}
Hanke, T., Movahedi, M., Williams, D.: Dfinity technology overview series,
  consensus system. arXiv preprint arXiv:1805.04548  (2018)

\bibitem{ouroboros}
Kiayias, A., Russell, A., David, B., Oliynykov, R.: Ouroboros: A provably
  secure proof-of-stake blockchain protocol. In: Annual International
  Cryptology Conference. pp. 357--388. Springer (2017)

\bibitem{byzcoin}
Kokoris{-}Kogias, E., Jovanovic, P., Gailly, N., Khoffi, I., Gasser, L., Ford,
  B.: Enhancing bitcoin security and performance with strong consistency via
  collective signing. CoRR  \textbf{abs/1602.06997} (2016),
  \url{http://arxiv.org/abs/1602.06997}

\bibitem{omniledger}
Kokoris-Kogias, E., Jovanovic, P., Gasser, L., Gailly, N., Ford, B.:
  Omniledger: A secure, scale-out, decentralized ledger. IACR Cryptology ePrint
  Archive  (2017), \url{https://eprint.iacr.org/2017/406.pdf}

\bibitem{zyzzyva}
Kotla, R., Alvisi, L., Dahlin, M., Clement, A., Wong, E.: Zyzzyva: speculative
  byzantine fault tolerance. In: ACM SIGOPS Operating Systems Review. vol.~41,
  pp. 45--58. ACM (2007)

\bibitem{tendermint}
Kwon, J.: Tendermint: Consensus without mining.
  \url{https://tendermint.com/static/docs/tendermint.pdf} (2014)

\bibitem{lamport}
Lamport, L., Shostak, R., Pease, M.: The byzantine generals problem. ACM
  Transactions on Programming Languages and Systems (TOPLAS)  \textbf{4}(3),
  382--401 (1982)

\bibitem{conflux}
Li, C., Li, P., Xu, W., Long, F., Yao, A.C.c.: Scaling nakamoto consensus to
  thousands of transactions per second. arXiv preprint arXiv:1805.03870  (2018)

\bibitem{libra}
Libra: \url{https://www.libra.org}

\bibitem{stellar}
Mazieres, D.: The stellar consensus protocol: A federated model for
  internet-level consensus. Stellar Development Foundation  (2015)

\bibitem{vrf1}
Micali, S., Rabin, M., Vadhan, S.: Verifiable random functions. In: 40th Annual
  Symposium on Foundations of Computer Science (Cat. No. 99CB37039). pp.
  120--130. IEEE (1999)

\bibitem{miller}
Miller, A., Xia, Y., Croman, K., Shi, E., Song, D.: The honey badger of {BFT}
  protocols. In: Proceedings of the 2016 ACM SIGSAC Conference on Computer and
  Communications Security. pp. 31--42. ACM (2016)

\bibitem{nakamoto}
Nakamoto, S.: Bitcoin: A peer-to-peer electronic cash system (2008),
  \url{https://bitcoin.org/bitcoin.pdf}

\bibitem{parity}
Parity: Proof-of-authority chains - wiki parity tech documentation.
  \url{https://wiki.parity.io/Proof-of-Authority-Chains}

\bibitem{hybrid}
Pass, R., Shi, E.: Hybrid consensus: Efficient consensus in the permissionless
  model. IACR Cryptology ePrint Archive  (2016),
  \url{http://eprint.iacr.org/2016/917.pdf}

\bibitem{plasma}
Poon, J., Buterin, V.: Plasma: Scalable autonomous smart contracts.
  \url{https://plasma.io/plasma.pdf} (2017)

\bibitem{lightning}
Poon, J., Dryja, T.: The bitcoin lightning network: Scalable off-chain instant
  payments. Technical Report (draft)  (2015),
  \url{https://lightning.network/lightning-network-paper.pdf}

\bibitem{tangle}
Popov, S.: The tangle. \url{https://iota.org/IOTA\_Whitepaper.pdf} (2014)

\bibitem{ren2}
Ren, Z., Cong, K., Pouwelse, J., Erkin, Z.: Implicit consensus: Blockchain with
  unbounded throughput. CoRR  \textbf{abs/1705.11046} (2017),
  \url{http://arxiv.org/abs/1705.11046}

\bibitem{avalanche}
Rocket, T.: Snowflake to avalanche: A novel metastable consensus protocol
  family for cryptocurrencies (2018)

\bibitem{phantom}
Sompolinsky, Y., Zohar, A.: Phantom: A scalable blockdag protocol (2018)

\bibitem{spectre}
Sompolinsky, Y., Lewenberg, Y., Zohar, A.: Spectre : Serialization of
  proof-of-work events : Confirming transactions via recursive elections (2017)

\bibitem{ghost}
Sompolinsky, Y., Zohar, A.: Secure high-rate transaction processing in bitcoin.
  In: International Conference on Financial Cryptography and Data Security. pp.
  507--527. Springer (2015)

\bibitem{vechain}
Vechain: \url{https://www.vechain.com}

\bibitem{ethereum}
Wood, G.: Ethereum: A secure decentralised generalised transaction ledger.
  Ethereum Project Yellow Paper  \textbf{151} (2014),
  \url{http://gavwood.com/paper.pdf}

\bibitem{polkadots}
Wood, G.: Polkadot: Vision for a heterogeneous multi-chain framework.
  \url{https://github.com/polkadot-io/polkadot-white-paper} (2016)

\bibitem{hotstuff}
Yin, M., Malkhi, D., Reiter, M.K., Gueta, G.G., Abraham, I.: Hotstuff: Bft
  consensus in the lens of blockchain. arXiv preprint arXiv:1803.05069  (2018)

\bibitem{ncmax}
Zhang, R.: Phd thesis: Analyzing and improving proof-of-work consensus
  protocols (2019)

\bibitem{metric}
Zhang, R., Preneel, B.: Lay down the common metrics: Evaluating proof-of-work
  consensus protocols' security. In: 2019 IEEE Symposium on Security and
  Privacy (SP). pp. 175--192. IEEE (2019)

\end{thebibliography}

\end{document}